\documentclass[10pt,conference]{IEEEtran}

\IEEEoverridecommandlockouts                              

\usepackage{graphics} 
\usepackage{epsfig} 
\usepackage{mathptmx} 
\usepackage{times} 
\usepackage{amsmath} 
\usepackage{amssymb}
\usepackage{url}
\usepackage{booktabs}
\usepackage{graphicx}
\usepackage{siunitx}
\usepackage{mathtools}
\usepackage{blindtext}
\usepackage{centernot}
\usepackage{algorithm}
\usepackage{algorithmic}
\usepackage{svg}
\usepackage{bbm}

\DeclareMathOperator*{\argmax}{arg\,max}
\DeclareMathOperator*{\argmin}{arg\,min}

\newtheorem{theorem}{Theorem}

\def\BibTeX{{\rm B\kern-.05em{\sc i\kern-.025em b}\kern-.08em
    T\kern-.1667em\lower.7ex\hbox{E}\kern-.125emX}}
    
\IEEEoverridecommandlockouts\IEEEpubid{\makebox[\columnwidth]{ 978-1-6654-3540-6/22~\copyright~2022 IEEE \hfill} \hspace{\columnsep}\makebox[\columnwidth]{ }}
    
\begin{document}

\title{Optimal Offloading Strategies for Edge-Computing via Mean-Field Games and Control\\
\thanks{This work has been co-funded by the LOEWE initiative (Hesse, Germany) within the emergenCITY center and the German Research Foundation (DFG) via the Collaborative Research Center (CRC) 1053 – MAKI.}
}

\makeatletter
\newcommand{\linebreakand}{%
  \end{@IEEEauthorhalign}
  \hfill\mbox{}\par
  \mbox{}\hfill\begin{@IEEEauthorhalign}
}
\makeatother

\author{\IEEEauthorblockN{ Kai~Cui, Mustafa~Burak~Yilmaz, Anam~Tahir, Anja~Klein, Heinz~Koeppl}
\IEEEauthorblockA{\textit{Department of Electrical Engineering and Information Technology} \\
\textit{Technische Universität Darmstadt}\\
Darmstadt, Germany \\
{\tt\small \{kai.cui, burak.yilmaz, anam.tahir, anja.klein, heinz.koeppl\}@tu-darmstadt.de}}
}

\maketitle

\begin{abstract}
The optimal offloading of tasks in heterogeneous edge-computing scenarios is of great practical interest, both in the selfish and fully cooperative setting. In practice, such systems are typically very large, rendering exact solutions in terms of cooperative optima or Nash equilibria intractable. For this purpose, we adopt a general mean-field formulation in order to solve the competitive and cooperative offloading problems in the limit of infinitely large systems. We give theoretical guarantees for the approximation properties of the limiting solution and solve the resulting mean-field problems numerically. Furthermore, we verify our solutions numerically and find that our approximations are accurate for systems with dozens of edge devices. As a result, we obtain a tractable approach to the design of offloading strategies in large edge-computing scenarios with many users. 
\end{abstract}


\begin{IEEEkeywords}
edge-computing, mean-field, computation offloading, Nash equilibria, Pareto optima
\end{IEEEkeywords}

\section{Introduction}
In recent years, a rapid growth of data generated from the network edge is witnessed, especially, the Cisco Annual Internet Report 2020 forecasts a rapid deployment of billions of Machine To Machine (M2M) devices until 2023 \cite{cisco20}. Multi-access Edge Computing (MEC) is a key technology to compensate strictly limited M2M devices in their processing by enabling computation offloading to cloudlet servers with computation resources in their vicinity. Additionally, the number of User Edge devices (UE) like smartphones, tablets and laptops also have increased tremendously due to the ease of their availability and low costs. These devices can also gain from offloading their tasks that demand intensive computations and low latencies, e.g., virtual reality, real-time face recognition, natural language processing, to such cloudlets. 

A MEC system can be seen as a multi-agent system where each UE is an agent who needs to decide, for each arriving task, whether to offload it to the MEC server or not. There has been great interest in finding the optimal policy for these UEs, to either offload or process locally, depending on factors such as their own available resources, network conditions and offloading computation costs. Even though several computation offloading strategies between UEs and MEC servers have been proposed in the literature, finding scalable solutions for MEC multiagent systems remains an important problem considering the continuously increasing number of agents.

Mean-field approximations are becoming an increasingly popular approach to resolve the curse of dimensionality in large-scale multi-agent systems \cite{zhang2021multi}. 
The idea is to represent a many-agent system as a single-representative-agent system which interacts with the empirical state distribution (or the mean-field) of all the other agents in the system, making the system tractable. 
It has been shown that the solution to a mean-field system is a good approximation to an $N$-agent system as $N$ grow large \cite{huang2006large}.
There are two main categories in which agents can work in a multi-agent system, either cooperatively to maximize a global goal or competitively to maximize their own reward, or a combination of both.
Mean-Field Games (MFGs, \cite{lasry2007mean, gueant2011mean, gomes2010discrete, saldi2018markov}) provide a way to analyze and solve large-scale competitive problems in a tractable manner, in particular under the usage of learning-based techniques \cite{guo2019learning, elie2020convergence}. On the other hand, Mean-Field Control (MFC) is used to model cooperative settings in many-agent systems \cite{andersson2011maximum, bensoussan2013mean}, see \cite{kizilkale2014collective, djehiche2016risk, lee2021controlling, fouque2020deep}, and references therein for some interesting applications. Similarly, we will formulate offloading in edge-computation as both a one-shot problem with theoretical guarantees and alternatively a time-stationary problem, allowing for competitive and cooperative solutions in a unified, tractable manner.

\begin{figure}
    \centering
    \includegraphics[width=0.95\linewidth]{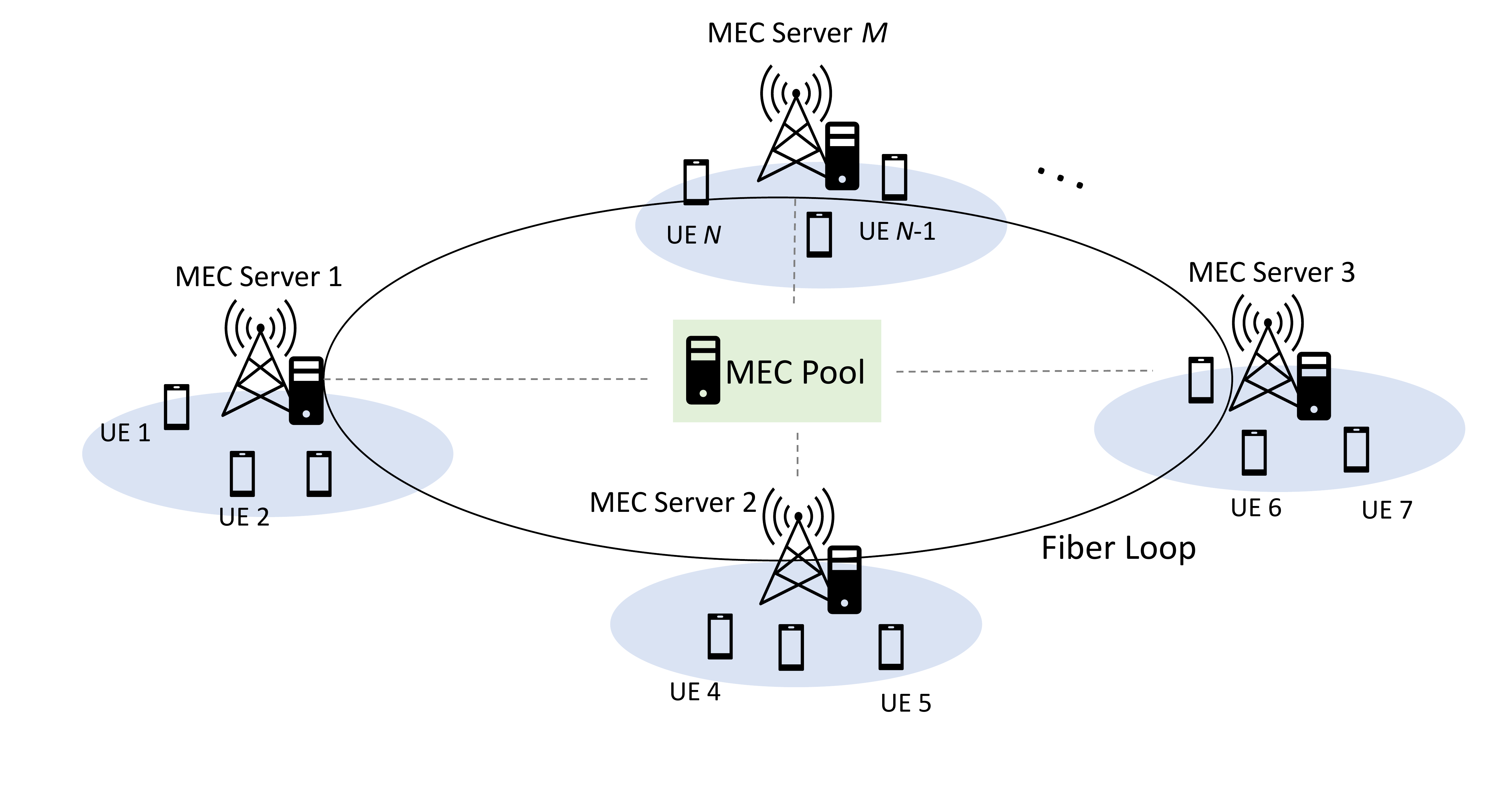}
    \caption{MEC scenario with $N$ UEs offloading their tasks to $M$ MEC servers, where computation resources of MEC servers are shared through a fiber loop connecting them, acting as a single processing pool.}
    \label{fig:system_model}
\end{figure}

The decision to offload tasks to MEC has been an important area of research since it is crucial to find a good, if not optimal, policy in scalable applications.
Various prior works have used mean-field approximations for similar offloading problems. In \cite{zheng2019optimal}, the authors model a shared MEC competitive offloading problem as an MFG in continuous time and solve the resulting coupled Hamilton-Jacobi-Bellman and Fokker-Plank partial differential equations using fixed-point iteration. However, their model considers only the non-cooperative case and results in a continuous time model, while in our work we also consider a cooperative setting and obtain a model for the time-stationary case.
Similarly, in \cite{banez2019mean, banez2020mean}, the authors consider both cooperative and non-cooperative computational offloading problems, though through the special case of a linear-quadratic model, while we solve a non-linear problem. 

Mean-field approximations have also been used to model large-scale MEC systems with D2D collaborations \cite{wang2021joint} particularly on graphs as a deterministic ODE system, though the model does not consider entirely selfish nodes. Finally, authors in \cite{hanif2015mean} model and solve a large-scale resource-sharing problem using mean-field theory and both cooperative and non-cooperative strategies, i.e. a more centralized setting without local computational capabilities. In contrast to our work, their models focus on graph-based job forwarding and continuous-action resource-sharing problems, whereas our model focuses on offloading decisions. For a variety of other works on applications of the mean-field approach in communication systems, see also e.g. \cite{gao2020mean, zheng2021dynamic, li2020resource, banez2021survey}. Apart from the discussed differences, all of the models in prior works and our work are diverse and apply to a variety of differing scenarios. Here, in contrast to previous work, our model will consider the both cooperative and competitive optimization of binary offloading decisions in edge-computing, where users may choose to offload or compute locally. In this work, we will formulate a unified mean-field framework for both the competitive and cooperative setting of offloading decisions in edge-computing. In particular, our tractable solution considers both a one-shot and time-stationary scenario that is rigorously approximating the finite user system. 

We begin our analysis with a computationally expensive-to-solve one-shot game, a knapsack problem, where many edge devices must independently decide whether to offload or not for a given distribution of task configurations. Our contribution can be summarized as follows: (i) We pass to the infinite-user limit in order to obtain a tractable problem with a complexity independent of the number of users; (ii) The model is theoretically motivated by showing novel existence and approximate optimality properties of solutions in large finite systems, both in terms of Nash equilibria and Pareto optima; (iii) Since, in practice, a one-shot game may not be sufficiently realistic, we extend our model to a new time-stationary model with a Poisson task arrival process and find analogous competitive and cooperative solutions in the limit of large systems; and finally (iv), the proposed models are verified in simulation and solved or learned with complexity independent of the number of users, while concurrently giving a good solution to large finite systems. As a result, we tractably solve an otherwise intractable many-agent offloading problem. 

\section{Mathematical Model}
In the considered scenario, a multi-cell ultra-dense network includes $M$ MEC servers and a total number of $N$ UEs associated with these MEC servers. It is assumed that $M$ MEC servers are connected through a fiber loop to share their computational resources, pooled into a single centralized but distributed MEC pool, where we also assume the bandwidth across the MEC servers is large enough for connecting all UEs and the delay in resource sharing through fiber loop is neglected, similar to the framework presented in~\cite{zheng2019optimal}. The scenario is depicted in Fig.~\ref{fig:system_model} with $M$ MEC servers and $N$ UEs. Each UE, modeled by $i=1,\ldots,N$, is given a configuration 
\begin{align}
    C^i \coloneqq (W^i, L^i, f^i, R^i) \in \mathbb K \subseteq \mathbb R_{\geq 0}^4
\end{align} 
with task transmission length $W^i$ in bits, task processing complexity $L^i$ in CPU cycles to be computed (either by offloading to a MEC server or by computing locally), transmission rate $R^i$ between user and MEC server measured in bits per second, and finally a local processing rate $f^i$ in CPU cycles per second. Each user $i$ may make a decision $X_i \in \{0, 1\}$ of whether to offload their task. 

For a given overall MEC server processing rate $f_{\mathrm{pool}} \in \mathbb R$, we assume each offloading UE is allocated a proportional amount of processing power from the MEC processing pool to their offloaded task's complexity. The time $T^{\mathrm{tx}}_i$ to transmit the $i$-th users' task to a MEC server and the time $T^{\mathrm{off}}_i$ to compute at the MEC server are given by
\begin{align}
    T^{\mathrm{tx}}_i = \frac{W^i}{R^i}, \quad
    T^{\mathrm{off}}_i = \frac{L^i}{f_{\mathrm{pool}}\cdot\frac{1}{\sum_j X_j}}
\end{align}
assuming that each offloading task is assigned equal processing power. Alternatively, one could also easily consider completion of all offloaded tasks at once, i.e. $T^{\mathrm{off}}_i = \frac{X_j \sum_j L^j}{f_{\mathrm{pool}}}$. The time to compute a task locally is given by 
\begin{align}
    T^{\mathrm{loc}}_i = \frac{L^i}{f^i}.
\end{align} 

The result of the computed tasks are assumed to be negligible in size compared to original task size $W^i$, therefore, the time needed for the reception of results is not considered.

\subsection{Competitive game setting}
In the competitive, selfish setting, each user $i$ independently decides on whether to offload or not via the random variable $X_i \in \{0, 1\}$ so as to minimize only their own expected computation time, i.e. either time to compute locally or offload
\begin{multline} \label{eq:exactcomp}
    \mathbb P (X_i = 1) \cdot (T^{\mathrm{tx}}_i + T^{\mathrm{off}}_i) + \mathbb P (X_i = 0) \cdot T^{\mathrm{loc}}_i \\
    = \mathbb E \left[ X_i (T^{\mathrm{tx}}_i + T^{\mathrm{off}}_i) + (1-X_i) T^{\mathrm{loc}}_i \right]
\end{multline}

Under full information, we obtain a standard static game with the classical solution concept of mixed Nash equilibria: Each of the users chooses whether to offload according to a policy $\pi_i$ which gives the conditional probability of offloading
\begin{align}
    \mathbb P(X_i = 1 \mid C^1, \ldots, C^N)
    \equiv \pi_i(C^1, \ldots, C^N)
\end{align}
which results in the minimization objective of each user $i$,
\begin{align}
    J_i^N({\pi_1}, \ldots, {\pi_N}) = \mathbb E \left[ X_i (T^{\mathrm{tx}}_i + T^{\mathrm{off}}_i) + (1-X_i) T^{\mathrm{loc}}_i \right].
\end{align}  

An approximate $\varepsilon$-Nash equilibrium is now defined as a tuple of policies $(\pi_1, \ldots, \pi_N)$ such that no user can gain by unilaterally changing their policy, i.e. for any $i = 1, \ldots, N$,
\begin{align}
    &J_i^N(\pi_i, \pi_{-i}) \leq \max_{\pi \in \Pi} J_i^N(\pi, \pi_{-i}) + \varepsilon
\end{align}
where $\pi_{-i}$ denotes all policies other than the $i$-th policy. Here, the minimal such $\varepsilon$ is also referred to as the exploitability of policies. An exact Nash equilibrium for $\varepsilon=0$ is indeed guaranteed to exist as long as $\mathbb K$ is compact (e.g. \cite{fudenberg1991game}). 

Unfortunately, it is known that the computation of Nash equilibria is hard, see \cite{daskalakis2009complexity}. Instead, we shall consider the many-agent case through mean-field analysis to obtain a tractable solution for large systems. At the same time, the solution will consist of decentralized policies. To this end, we shall now assume that there exists an underlying distribution $\mu_0 \in \mathcal P(\mathbb K)$ of user specifications, i.e. for $i=1,\ldots,N$ we have random variables $C^i = (W^i, L^i, f^i, R^i) \sim \mu_0$. To obtain a reasonable solution, we must also assume that the MEC pool computing power scales suitably with the number of users, i.e. $f_\mathrm{pool} = N \cdot f_\mathrm{per}$ for some $f_\mathrm{per} \in \mathbb R$, since otherwise in the limit of many agents, offloading will become pointless. In practice, for fixed $f_\mathrm{pool}$ and $N$ in given finite $N$-agent systems, this may be realized by defining $f_\mathrm{per} \coloneqq \frac{f_\mathrm{pool}}{N}$.

We now consider a decentralized control setting by allowing each agent to decide whether to offload depending only on their own configuration $C_i$. For motivation, note that since all other agents are exchangeable from the perspective of a single agent, only the own state and overall distribution of behaviors of other agents matters. Furthermore, in the limit of $N \to \infty$, the other users' distribution is uninformative, since under a common offloading strategy, the distribution converges to some fixed mean-field by the law of large numbers. Additionally, decentralized control policies may be motivated in practice by limited agent information. Since the computation of Nash equilibria in this setting nonetheless remains hard, this motivates the mean-field formulation.

For tractability, we formulate a mean-field game as $N \to \infty$, as popularized by \cite{huang2006large} and \cite{lasry2007mean} for stochastic differential games. Here, we propose a mean-field model with near-Nash properties as $N$ grows large, as we will also verify theoretically. Consider a policy $\pi$ shared by all users. The policy induces a joint distribution $\mu = \mu_0 \otimes \pi$ over user states and offloading decisions. Under this fixed distribution $\mu$, the objective of a single, representative user becomes 
\begin{align}
    &J^\mu({\pi}) = \mathbb E \left[ X (\tilde T^{\mathrm{tx}} + \tilde T^{\mathrm{off}}) + (1-X) \tilde T^{\mathrm{loc}} \right]
\end{align}
where we have expectations of random variables of the representative agent $(W, L, f, R, X) \sim \mu \otimes \pi$, and random transfer or processing times of the mean-field system $W, L, f, R$
\begin{align}
    \tilde T^{\mathrm{tx}} = \frac{W}{R}, \quad \tilde T^{\mathrm{off}} = \frac{L \int X \, \mathrm d\mu}{f_{\mathrm{per}}}, \quad \tilde T^{\mathrm{loc}} = \frac{L}{f}.
\end{align}

The $N \to \infty$ analogue of Nash equilibria is the mean-field equilibrium, defined as a tuple $(\pi^*, \mu^*)$ of policy and mean-field, such that the policy is optimal under the mean-field generated by itself, i.e. defined through the fixed point equation
\begin{subequations} \label{eq:MFG}
    \begin{alignat}{2}
        \pi^* &= \argmin_\pi J^{\mu^*}({\pi}), \\
        \mu^* &= \mu_0 \otimes \pi^*.
    \end{alignat}
\end{subequations}
Analytically, for any fixed mean-field $\mu$, we could find such a best response policy $\operatorname{BR}(\mu)$ by defining
\begin{align}
    \pi^*(W, L, f, R) = \mathbf 1_{\tilde T^{\mathrm{tx}}(W, L, f, R) + \tilde T^{\mathrm{off}}(W, L, f, R) < \tilde T^{\mathrm{loc}}(W, L, f, R)}.
\end{align}
However, simply iterating the two fixed point equations is generally not guaranteed to converge to an equilibrium. Thus, we will learn equilibria through fictitious play \cite{elie2020convergence}.

\subsection{Cooperative control setting}
In contrast to the selfish, competitive setting, in a cooperative setting it may be of interest to minimize the average processing time of all users. One may formulate a centralized optimization problem as
\begin{subequations} \label{eq:exactcoop}
    \begin{align}
    &\! \min_{X_1, \ldots, X_N} &\quad& \frac 1 N \sum_{i=1}^{N} X_i (T^{\mathrm{tx}}_i + T^{\mathrm{off}}_i) + (1-X_i) T^{\mathrm{loc}}_i \\
    &\text{subject to} & & X_i \in \{0, 1\} \qquad \forall i \in \{ 1,\ldots,N \} 
    \end{align}
\end{subequations}
under full information. However, again this problem is known to be difficult to solve exactly for large $N$, as it is a knapsack problem \cite{9322420}. Furthermore, we may again be interested in a decentralized solution, where each agent uses an independent policy, eliminating the need for centralized knowledge and only requiring knowledge of the local configuration $C^i$. Reformulating as optimization over decentralized policies $\pi_i$ and optimizing over the expected cost, we have
\begin{subequations} \label{eq:exactcooppi}
    \begin{align}
    &\! \min_{\pi^1, \ldots, \pi^N} &\quad& \mathbb E \left[ \frac 1 N \sum_{i=1}^{N} X_i (T^{\mathrm{tx}}_i + T^{\mathrm{off}}_i) + (1-X_i) T^{\mathrm{loc}}_i \right]\\
    &\text{subject to} & & \pi_i \colon \mathbb K \to [0,1] \qquad \forall i \in \{ 1,\ldots,N \} 
    \end{align}
\end{subequations}
where each offloading decision $X_i \sim \mathrm{Bernoulli}(\pi_i(C_i))$ follows from the policy $\pi_i$.

As $N \to \infty$, under the policy $\pi$ for all agents, we can obtain the corresponding mean-field control problem, which is more tractable than directly solving the $N$-user system, given as
\begin{align}
&\! \min_{\pi} &\quad& \int X (\tilde T^{\mathrm{tx}} + \tilde T^{\mathrm{off}}) + (1-X) \tilde T^{\mathrm{loc}} \, \mathrm d(\mu_0 \otimes \pi)
\end{align}
again with the previous definitions. Note that although we impose a shared, common policy $\pi$, sharing a policy across all agents will indeed be sufficient for optimality \cite{9683749}.

Since the problem is now reduced to the choice of $\pi \colon \mathbb K \to [0,1]$, the combinatorial optimization problem has been reduced to optimization over a bounded function $\pi$ with complexity independent of $N$. If we further assume that $\mu_0$ has finite support, i.e. $K \coloneqq |\mathbb K| < \infty$ and
\begin{align}
    \mu_0 = \sum_{j=1}^K p_j \delta_{(W_j, L_j, f_j, R_j)}, \quad \sum_{j=1}^K p_j = 1,
\end{align}
for some $p_j \geq 0$, $(W_j, L_j, f_j, R_j) \in \mathbb K$, then we obtain
\begin{align*}
    &\int X (\tilde T^{\mathrm{tx}} + \tilde T^{\mathrm{off}}) + (1-X) \tilde T^{\mathrm{loc}} \, \mathrm d(\mu_0 \otimes \pi) \\
    &= \sum_{j=1}^K p_j \pi_j \left( \frac{W_j}{R_j} + \frac{L_j \sum_{k=1}^K p_k \pi_k}{f_{\mathrm{per}}} \right) + p_j (1-\pi_j) \frac{L_j}{f_j} \\
    &= \sum_{j=1}^K \sum_{k=1}^K \frac{p_j p_k L_j}{f_{\mathrm{per}}} \pi_j \pi_k + \sum_{j=1}^K \left( \frac{p_j W_j}{R_j} - \frac{p_j L_j}{f_j} \right) \pi_j + \sum_{j=1}^K  \frac{p_j L_j}{f_j} \\
    &= \boldsymbol \pi^T \boldsymbol Q \boldsymbol \pi + \boldsymbol c^T \boldsymbol \pi + \mathrm{const.}
\end{align*}
for $\boldsymbol \pi \equiv (\pi_1, \ldots, \pi_K)^T$, $\pi_j = \pi(W_j, L_j, f_j, R_j)$ and appropriate $\boldsymbol Q$, $\boldsymbol c$. Therefore, we obtain a non-convex quadratic program
\begin{subequations} \label{eq:MFC}
    \begin{align}
    &\! \min_{\pi_1, \ldots, \pi_K} &\quad& \boldsymbol \pi^T \boldsymbol Q \boldsymbol \pi + \boldsymbol c^T \boldsymbol \pi \\
    &\text{subject to} & & \pi_j \in [0, 1] \qquad \forall j \in \{ 1,\ldots,K \}
    \end{align}
\end{subequations}
with box constraints, which though NP-hard \cite{pardalos1991quadratic} in the cardinality of the support of $\mu_0$, can be solved numerically. Most importantly, the complexity remains independent of $N$, giving us a tractable solution for sufficiently small $K$. To handle more general densities $\mu_0$ with non-finite but compact support $\mathbb K$, we may discretize distributions and solve the resulting finite-support problem. As a result, we have obtained a tractable solution to the otherwise intractable offloading problem for many devices, as we will verify in the sequel.

\section{Time-Stationary Model} \label{sec:time}
While the previous model assumes an instantaneous problem where we let all users play a one-shot game, another important and interesting setting is to assume a continuous flow of tasks arriving over time. While a theoretically rigorous analysis of this setting is beyond the scope of our work, we nonetheless consider this setting at its time-stationary equilibrium and solve it numerically.

At all times, let the arrival process of tasks be given by a Poisson process with constant rate $\lambda N$, which is equivalent to Poisson arrival rates $\lambda$ for each of $N$ users. At equilibrium, in the limit there must be a time-stationary bandwidth per user $f_\mathrm{alloc}$ allocated to a user choosing to offload their task. This bandwidth is given by dividing the total processing power $f_\mathrm{pool} = N f_\mathrm{per}$ by the number of jobs in the system. Since the processing time for any offloaded job arriving at equilibrium is given by $T^{\mathrm{tx}} = \frac{W}{R}$ and $T^{\mathrm{off}} = \frac{L}{f_{\mathrm{alloc}}}$, the expected number of jobs in the system as $N \to \infty$ will be given by
\begin{align}
    \mathbb E \left[ N_\mathrm{tot} \right] = \lambda N \mathbb E \left[ X (T^{\mathrm{tx}} + T^{\mathrm{off}}) \right]
\end{align}
and is given by a sum of $N$ Poisson variables $N_\mathrm{tot}^i$, the numbers of jobs in the system from each user $i$. Therefore, by the central limit theorem, the fluctuations of $N_\mathrm{tot}$ are on the order of $O(\sqrt N)$, resulting in the allocated processing rate per user
\begin{multline}
    f_\mathrm{alloc} 
    = \frac{f_\mathrm{pool}}{\mathbb E \left[ N_\mathrm{tot} \right] + O(\sqrt N)}
    = \frac{f_\mathrm{per}}{\lambda \mathbb E \left[ X (T^{\mathrm{tx}} + T^{\mathrm{off}}) \right] + O(\frac{1}{\sqrt N})} \\
    \to \frac{f_\mathrm{per}}{\lambda \mathbb E \left[ \frac{XW}{R} \right] + \frac{\lambda \mathbb E \left[ XL \right]}{f_{\mathrm{alloc}}}}
\end{multline}
as $N \to \infty$, which for $f_\mathrm{alloc} \neq 0$ gives
\begin{align}
    f_\mathrm{alloc} = \frac{f_\mathrm{per} - \lambda \mathbb E \left[ XL \right]}{\lambda \mathbb E \left[ \frac{XW}{R} \right]}
\end{align}
and the natural constraint 
\begin{align}
    f_\mathrm{per} - \lambda \int XL \, \mathrm d(\mu_0 \otimes \pi) > 0.
\end{align}

Intuitively, this constraint formalizes the notion of sufficient MEC resources, i.e. the rate of assigned jobs times their complexity must not exceed the possible compute assigned per node, as otherwise the MEC servers will be unable to catch up with assigned tasks, resulting in no time-stationary solution. Note that this constraint is trivially fulfilled if
\begin{align}
    f_\mathrm{per} > \lambda \mathbb E \left[ L \right].
\end{align}

Optimizing the average waiting times of all agents in the cooperative case gives the MFC problem
\begin{align} \label{eq:MFC-homog}
&\! \min_{\pi \in [0, 1]^{\mathbb K}} \mathbb E \left[ X (T^{\mathrm{tx}} + T^{\mathrm{off}}) + (1-X) \tilde T^{\mathrm{loc}} \right]
\end{align}
where for finite $\mathbb K$ we have
\begin{align*}
    &\mathbb E \left[ X (T^{\mathrm{tx}} + T^{\mathrm{off}}) + (1-X) \tilde T^{\mathrm{loc}} \right] \\
    &\quad = \mathbb E \left[ \frac{XW}{R} + \frac{XL \lambda \mathbb E \left[ \frac{XW}{R} \right]}{f_\mathrm{per} - \lambda \mathbb E \left[ XL \right]} + (1-X) \frac{L}{f} \right] \\
    &\quad = \sum_{j=1}^K p_j \pi_j \left( \frac{W_j}{R_j} + \frac{L_j \lambda \sum_{k=1}^K \frac{p_k \pi_k W_k}{R_k}}{f_{\mathrm{per}} - \lambda \sum_{k=1}^K p_k \pi_k L_k} \right) + p_j (1-\pi_j) \frac{L_j}{f_j}.
\end{align*}

For the competitive MFG, we can analogously define an equilibrium as any fixed point policy $\pi^*$ such that 
\begin{align} \label{eq:MFG-homog}
    \pi^* \in \argmax_{\pi} \mathbb E_\pi \left[ \frac{XW}{R} + \frac{XL \lambda \mathbb E_{\pi^*} \left[ \frac{XW}{R} \right]}{f_\mathrm{per} - \lambda \mathbb E_{\pi^*} \left[ XL \right]} + (1-X) \frac{L}{f} \right].
\end{align}

\section{Theoretical Guarantees} \label{sec:theory}
In this section, we state a number of theoretical guarantees for the one-shot mean-field problems. Extensions to the time-stationary case are deferred to future work. The results follow from formulating the problem as certain standard mean-field game and mean-field control problems and applying existing results. In particular, note that our systems can be reformulated as standard mean-field games with action space $\{0, 1\}$ and state space $\mathbb K \cup (\mathbb K \times \mathbb U)$, see also \cite{saldi2018markov, 9683749}. For the competitive setting, as $N \to \infty$, the MFG equilibrium exists and will constitute an approximate Nash equilibrium.
\begin{theorem}
A solution $(\pi^*, \mu^*)$ of the MFG problem \eqref{eq:MFG} exists, and $\pi^*$ constitutes an $\epsilon_N$-Nash equilibrium of the finite $N$-user system with $\epsilon_N \to 0$ as $N \to \infty$.
\end{theorem}
\begin{proof}
See \cite[Theorem~4.1]{saldi2018markov}.
\end{proof}

Furthermore, it is known that the fictitious play algorithm will converge in terms of exploitability, giving us the desired approximate Nash equilibrium.
\begin{theorem}
The exploitability of the solution of the fictitious play algorithm converges to zero.
\end{theorem}
\begin{proof}
The system fulfills \cite[Assumption 1]{elie2020convergence} and in particular the monotonicity property, since the offloading cost only increases when more agents offload. Therefore, by \cite[Corollary 8.2]{elie2020convergence}, we have convergence of the fictitious play algorithm to the unique mean-field equilibrium.
\end{proof}

Similarly, the cooperative MFC solution has an optimal solution, which will constitute an approximate Pareto optimum in the finite user system.
\begin{theorem}
For distributions $\mu_0$ with finite support, an optimizer $\pi^*$ of \eqref{eq:MFC} exists, and $\pi^*$ constitutes an $\epsilon_N$-Pareto optimum of the finite $N$-user system with $\epsilon_N \to 0$ as $N \to \infty$.
\end{theorem}
\begin{proof}
Existence is trivially guaranteed by the extreme value theorem, since the objective is a continuous function of $\pi \in [0, 1]^K$, and $[0, 1]^K$ is compact. For approximate Pareto-optimality, see \cite[Corollary~1]{9683749}.
\end{proof}

\begin{figure}[t]
    \centering
    \vspace{0.05cm}
    \includegraphics[width=\linewidth]{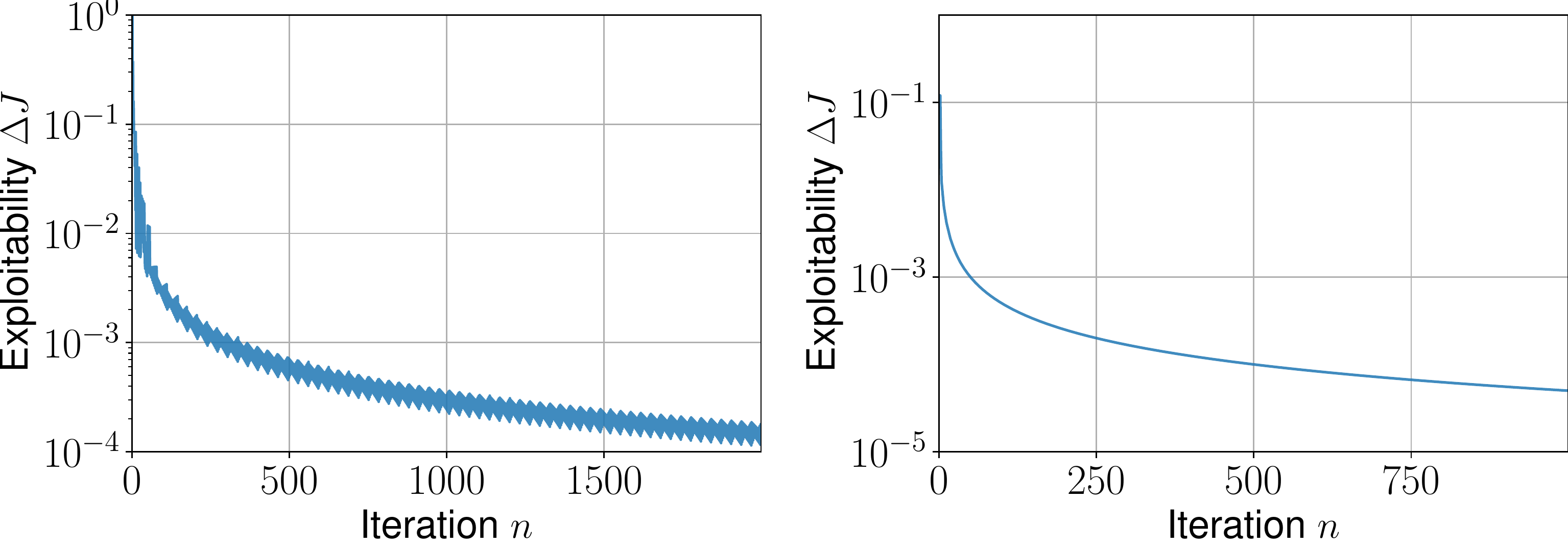}
    \caption{Learning curve for the exploitability $\Delta J$ in the competitive MFG problem \eqref{eq:MFG} (left) over $5000$ iterations $n$ using fictitious play. The fictitious play algorithm quickly converges to the equilibrium $\pi^* \approx (1, 0.65, 0)$. Here, we used $(p_1, p_2, p_3) = (0.2, 0.4, 0.4)$, $f_\mathrm{per} = 0.5$ and $(C_i)_{i=1,2,3} = ((1, 1, 1, 20), (3, 2, 1, 20), (5, 3, 1, 20))$. Similar results are achieved in the time-stationary MFG problem \eqref{eq:MFG-homog} (right), converging to the equilibrium $\pi^* \approx (1, 0.65, 0)$ for $(p_1, p_2, p_3) = (0.2, 0.4, 0.4)$, $f_\mathrm{per} = 0.5$, $\lambda \approx 0.225$ and $(C_i)_{i=1,2,3} = ((1, 1, 5, 10), (3, 2, 5, 10), (5, 3, 5, 10))$.}
    \label{fig:comp}
\end{figure}

\begin{figure}[b]
    \centering
    \includegraphics[width=\linewidth]{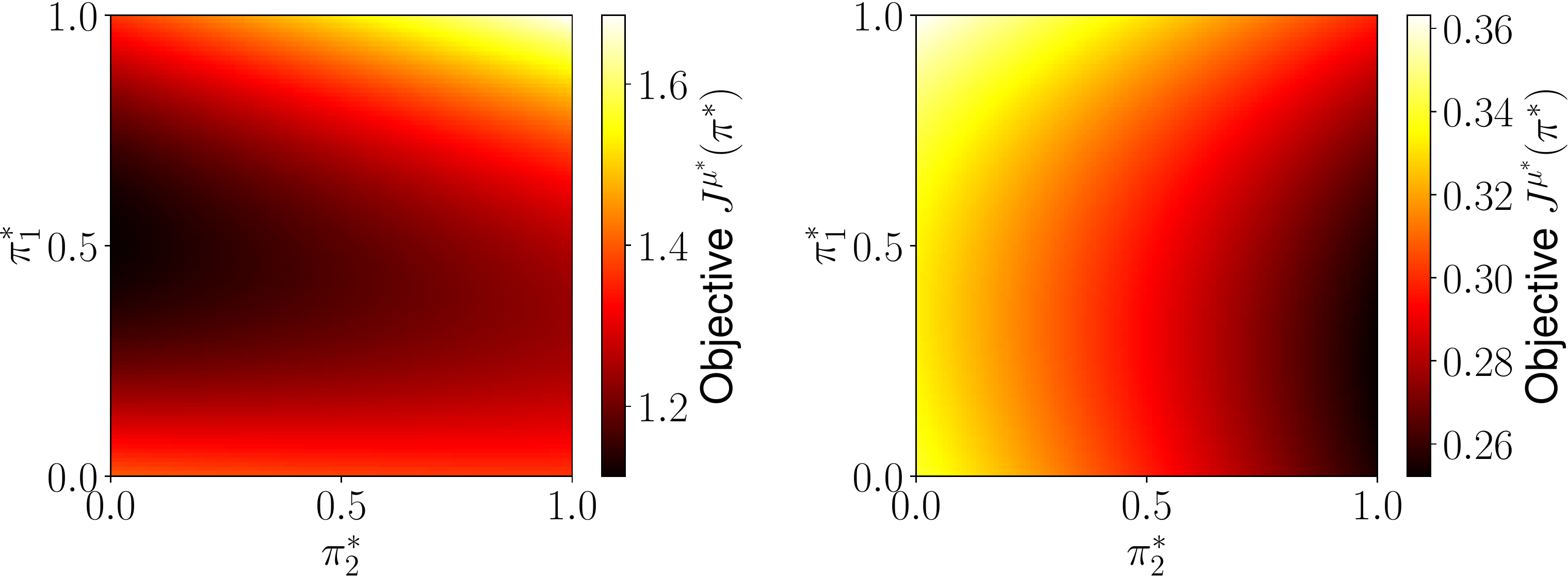}
    \caption{Exemplary 2D-case for the quadratic MFC problem \eqref{eq:MFC} (left) and $K=2$, reaching the optimal objective of around $1.26$ at around $\pi^* \approx (0.52, 0)$. Here, we used $(p_1, p_2) = (0.8, 0.2)$, $f_\mathrm{per} = 3$ and $(C_i)_{i=1,2} = ((3, 5, 3, 10), (1.5, 1.5, 5, 25))$. Similar results are achieved for the time-stationary MFC problem \eqref{eq:MFC-homog} (right), where we achieve the optimal value $0.25$ at $\pi^* \approx (0.24, 1)$ for $(p_1, p_2) = (0.8, 0.2)$, $f_\mathrm{per} = 3$, $\lambda = 0.6$ and $(C_i)_{i=1,2} = ((3, 1.5, 5, 12), (1.5, 1, 2, 20))$.}
    \label{fig:coop}
\end{figure}

\section{Numerical Simulation}
In this section, we present numerical simulations for the systems established in the prequel. For the quadratic program MFC, we could apply convex quadratic program solvers if the problem is indeed convex. However, in general the MFC problem may be non-convex and thus results in a NP-hard problem \cite{pardalos1991quadratic}. Still, we again stress that the complexity scales only with the size of $\mathbb K$ and remains independent of the number of users $N$. Therefore, our formulation will be of lower complexity than solving the finite user model for large systems. For space reasons, we do not compare run times, since finite model solvers will trivially exceed the run time of our solution for sufficiently large systems. We may follow any global optimization algorithm, and for simplicity we apply a simple grid search, though more sophisticated algorithms such as Bayesian optimization can easily be substituted. 

\begin{figure}[t]
    \centering
    \vspace{0.05cm}
    \includegraphics[width=\linewidth]{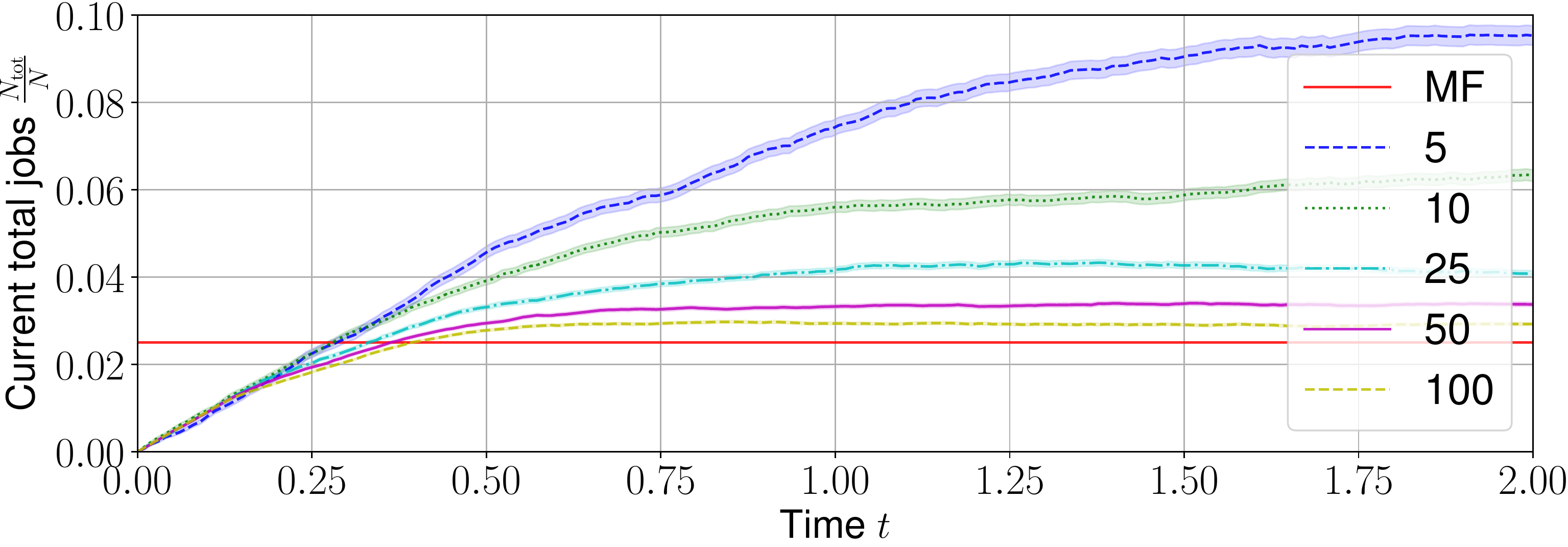}
    \caption{The evolution of the expected total number of jobs in the queue divided by $N$ with $68\%$ confidence interval, plotted for the configuration from Fig.~\ref{fig:comp} and various $N$ from $5$ to $100$, compared against the stationary mean-field solution (MF). We average over $5000$ sample trajectories. As we consider increasingly large systems, the expected rescaled number of jobs in the system converges to the limiting mean-field description, letting us conclude that the limiting mean-field system is a good approximation for the finite user system.}
    \label{fig:ergodic}
\end{figure}

\begin{figure}[t]
    \centering
    \includegraphics[width=\linewidth]{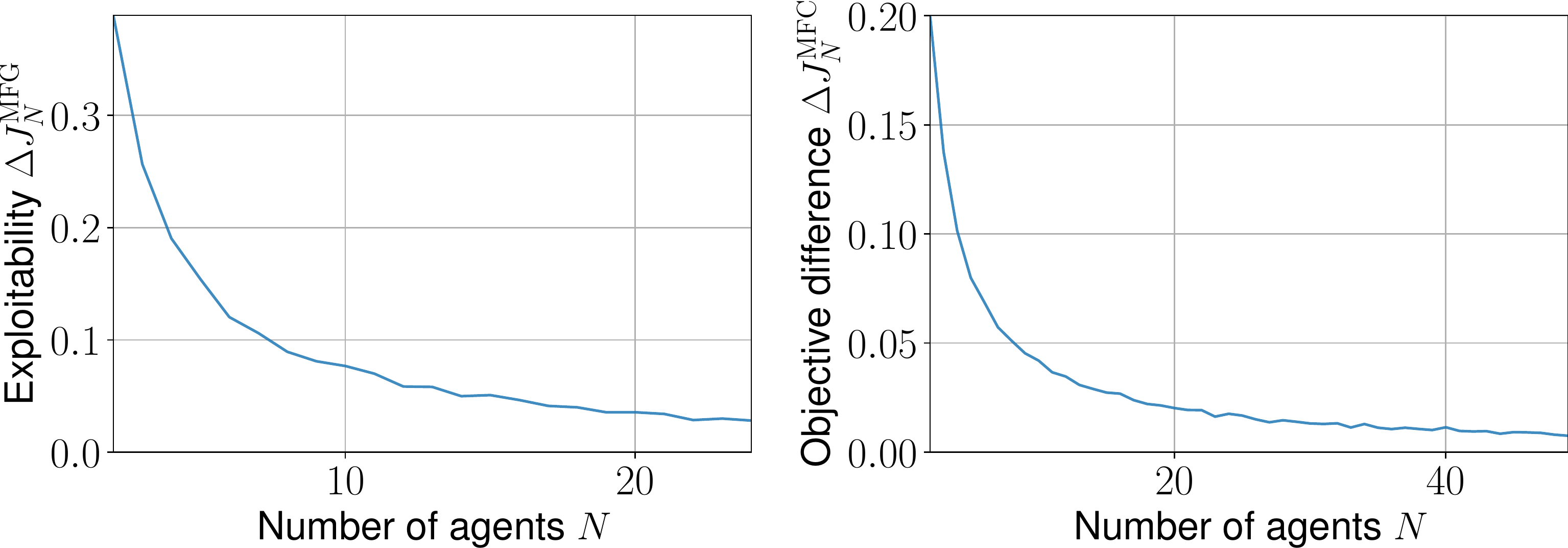}
    \caption{Comparison of $N$-agent exploitability in the MFG \eqref{eq:MFG} (left), and $N$-agent objective deviation from the limiting objective in the MFC \eqref{eq:MFC} (right) for the configurations used in Figures~\ref{fig:comp} and \ref{fig:coop}. The exploitability quickly decreases to zero, and similarly the cooperative problem is quickly well-approximated by the MFC.}
    \label{fig:coopncomp}
\end{figure}

As can be observed in Fig.~\ref{fig:comp}, for the competitive MFG problem \eqref{eq:MFG}, the fictitious play algorithm 
\begin{align}
    \pi_{n+1} \equiv \frac 1 {n+1} \left( n \pi_{1:n} + \argmin_\pi J^{\mu_0 \otimes \pi_{1:n}}({\pi}) \right)
\end{align} 
with the past average policy $\pi_{1:n} \coloneqq \frac 1 n \sum_{m=1}^n \pi_m$ quickly converges in terms of the exploitability
\begin{align}
    \Delta J(\pi) \coloneqq \max_{\pi^*} J^{\mu_0 \otimes \pi}({\pi^*}) - J^{\mu_0 \otimes \pi}({\pi})
\end{align}
which must be equal to zero for an exact equilibrium. For the parameters given in Fig.~\ref{fig:comp}, the resulting equilibrium $\pi^* \approx (0, 0.81, 1)$ is intuitive: Only the second configuration splits between offloading and local computation at a ratio that equilibrates offloading and local computation time, since the second configuration has longer offloading times than the first, and longer local computation time than the third. The result are offloading decisions where each UE gains little by deviating.

In Fig.~\ref{fig:coop}, we can observe the cost function for an illustrative case where $K=2$. As can be seen in the example, the optimum computational offloading policy is reached at around $\pi^* \approx (0.52, 0)$. Similarly, a solution can be reached for the time-stationary problem at around $\pi^* \approx (0.24, 1)$. Here, we solve the problem for an illustrative 3D example in a few seconds, though similar results can easily be obtained for larger problems. Thus, we obtained nearly optimal offloading decisions, minimizing the average computation times.

In Fig.~\ref{fig:ergodic}, we can observe that the time-homogeneous problem empirically shows a number of jobs in the system that converges to the mean-field description when rescaled by $N$, leading us to the conclusion that the mean-field model we proposed is a good approximation to the finite user system as long as the system is sufficiently large.

Finally, in Fig.~\ref{fig:coopncomp}, we can observe (i) the exploitability in the competitive finite user system, i.e. the expected maximum gain by deviating to any other policy in \eqref{eq:exactcomp}, and (ii) the deviation of the objective \eqref{eq:exactcooppi} from the computed mean-field objective \eqref{eq:MFC} in the cooperative setting. Here we estimated the exploitability for each value of $N$ by taking the maximum over all pure policies $\pi \in \{0, 1\}^{\mathbb K}$ over $100 \,\, 000$ samples. Similarly, we estimated the deviation between \eqref{eq:exactcooppi} and \eqref{eq:MFC} over $20 \,\, 000$ samples of the finite user system. We observe that the exploitability and deviation of objectives tends to zero as the number of agents increases, showing that the mean-field solution solves the finite system well.

\section{Conclusion}
In this work, we have shown the general applicability of rigorous mean-field frameworks for both competitive and cooperative scenarios in offloading for edge-computing. In particular, we have shown that the mean-field approximation quickly becomes a good approximation and can reliably be solved with a complexity independent of the number of agents. As a result, we have obtained good and tractable solutions for large-scale, decentralized edge-computing systems. In future work, one could extend rigorous theoretical analysis to the time-stationary case. Other interesting directions could be an extension to Markov-modulated task arrival rates and thereby a non-time-stationary case, or an even more distributed setting with multiple separate limited-access MEC pools. Finally, an application to real systems may be of interest.



\bibliographystyle{IEEEtran}
\bibliography{IEEEabrv,references}

\end{document}